\documentclass[letterpaper, 10 pt, journal, twoside]{IEEEtran} 
\IEEEoverridecommandlockouts  

\usepackage{graphics} 
\usepackage{epsfig} 
\usepackage{mathptmx} 
\usepackage{amsmath} 
\usepackage{amssymb} 
\allowdisplaybreaks 
\usepackage{cite}
\usepackage{relsize} 
\usepackage{xcolor}
\usepackage{graphicx}
\usepackage{amsfonts}
\usepackage{mathtools}
\usepackage[mathcal]{euscript}
\usepackage{mathtools, cuted} 
\usepackage[ruled,vlined]{algorithm2e}
\usepackage{enumitem} 
\usepackage{balance} 
\usepackage{mathrsfs} 
\usepackage{fdsymbol} 
\usepackage{soul} 
\usepackage{subcaption}

\DeclareCaptionLabelSeparator{periodspace}{.\quad}
\usepackage{amsthm}

\newtheorem{lemma}{Lemma}

\newtheorem{theorem}{Theorem}
\newtheorem{corollary}{Corollary}

\theoremstyle{definition}
\newtheorem{example}{Example}

\newtheorem{definition}{Definition}

\title{
Control Design for Risk-Based Signal Temporal Logic Specifications
}
\author{Sleiman Safaoui, Lars Lindemann, Dimos V. Dimarogonas, Iman Shames, Tyler H. Summers
\thanks{This material is based on work supported by the United States Air Force Office of Scientific Research under award number FA2386-19-1-4073. This work was also supported in part by the Swedish Research Council (VR) and Defence Science and Technology Group, Australian Government, through agreement MyIP: ID9156.}
\thanks{S. Safaoui is with the department of Electrical Engineering and T. H. Summers is with the department of Mechanical Engineering at University of Texas at Dallas, Richardson, TX, USA, E-mail: \{sleiman.safaoui, tyler.summers\}@utdallas.edu. L. Lindemann and D. V. Dimarogonas are with the Division of Decision and Control Systems, KTH Royal Institute of Technology, 100 44 Stockholm, Sweden, E-mail: \{llindem, dimos\}@kth.se. I. Shames is with the University of Melbourne, Parkville, VIC 3010, Australia, E-mail: iman.shames@unimelb.edu.au. 
}
}

\begin{document}

\maketitle
\thispagestyle{empty} 

\begin{abstract}
We present a general framework for risk semantics on Signal Temporal Logic (STL) specifications for stochastic dynamical systems using axiomatic risk theory. We show that under our recursive risk semantics, risk constraints on STL formulas can be expressed in terms of risk constraints on atomic predicates. We then show how this allows a (\emph{stochastic}) STL risk constraint to be transformed into a risk-tightened \emph{deterministic} STL constraint on a related deterministic nominal system, enabling the application of existing STL methods. For affine predicate functions and a (coherent) Distributionally Robust Value at Risk measure, we show how risk constraints on atomic predicates can be reformulated as tightened deterministic affine constraints. We demonstrate the framework using a Model Predictive Control (MPC) design with an STL risk constraint.
\end{abstract}


\section{Introduction} \label{sec:int}
\IEEEPARstart{T}EMPORAL logics allow to reason about temporal properties of systems and have traditionally been used in formal verification and model checking \cite{baier}. More recently, temporal logics have also been used to impose highly expressive mission specifications on complex autonomous systems. For systems under linear temporal logic (LTL) and metric interval temporal logic (MITL) specifications, motion planning and control synthesis algorithms have been proposed in \cite{kress2009temporal,guo2015multi,kloetzer2008fully}.

\emph{Signal Temporal Logic (STL)} is a temporal logic interpreted over dense-time real-valued \emph{(deterministic)} signals \cite{maler2004monitoring} similar to MITL; it allows to additionally impose quantitative spatial properties by means of predicates that go beyond the abstract use of propositions in LTL and MITL. STL is hence more expressive and has been the focus of motion planning and control synthesis in areas such as robotics. In addition to the Boolean satisfaction relation given for LTL and MITL specifications, one can associate quantitative semantics with an STL specification that allow to reason about how robustly (severely) a specification is satisfied (violated). These quantitative semantics come in the form of the robustness degree and the robust semantics \cite{fainekos2009robustness} as well as space robustness, time robustness, and other variants \cite{donze2010robust, mehdipour2019average,lindemann2019robust}. 

Control of systems under STL specifications is inherently different compared to systems under LTL and MITL specifications, which is mainly based on abstractions and automata theory. For deterministic discrete-time systems, the authors in \cite{raman2014model} transform the STL specification into mixed-integer linear constraints and use Model Predictive Control (MPC) to deal with these constraints. Similarly, MPC has been employed by maximizing certain forms of the quantitative semantics associated with the STL specification at hand \cite{mehdipour2019average,lindemann2019robust,pant2018fly}. These methods seem computationally more tractable than the approach presented in \cite{raman2014model} due to the use of smooth quantitative semantics. For deterministic continuous-time systems, \cite{lindemann2019efficient} uses timed-automata theory to decompose the STL specification into STL subspecifications that can be implemented by low-level feedback control laws, such as those in \cite{lindemann2018control,garg2019control,srinivasan2018control}. Learning-based methods for partially unknown systems have been presented in \cite{varnai2019prescribed,muniraj2018enforcing,xu2018advisory}.

A large fraction of the aforementioned research focuses on finite abstractions and/or \emph{deterministic} systems. However, methodological advances are required to account for inherent uncertainties and high-dimensional continuous spaces in autonomous systems, especially due to \emph{stochastic uncertainties} arising from the use of noisy data and learning components. Robust extensions of \cite{raman2014model} have been presented in \cite{sadraddini2015robust,raman2015reactive}. Probabilistic notions of STL for stochastic systems have been presented in
\cite{sadigh2016safe,farahani2018shrinking}. However, these formulations are either deterministic and based on a worst-case approach or utilize chance constraints, an incoherent measure of risk that can lead to undesirable decisions. 
Effective risk management in complex autonomous systems demands a more sophisticated approach to quantifying risks of specification violations. This motivates an axiomatic approach to risk, which has been advocated for in finance \cite{artzner1999coherent} and more recently in robotics \cite{majumdar2020should}. Control under coherent risk measures has been considered in \cite{samuelson2018safety,singh2018framework}. The extension to more complex and generic specifications, such as those captured by temporal logics, has however not been addressed.

\textbf{Contributions.}
1) We present a general framework for defining risk semantics of STL specifications for \emph{stochastic dynamical systems} using axiomatic risk theory. In particular, we compose risk metrics with predicate functions, which become stochastic in the considered setup. 
2) We then recursively define risk semantics for Boolean and temporal STL operators. For a given STL specification, we show that these risk semantics can be expressed as risk constraints on predicate functions over certain time intervals. We then show how this allows such a risk constraint to be transformed into a risk-tightened deterministic STL constraint on a related deterministic nominal system.
3) For affine predicates and a coherent Distributionally Robust Value at Risk measure, we show how risk constraints on predicate functions can be explicitly reformulated as tightened deterministic affine constraints. 
4) To demonstrate the framework, we use an MPC formulation to solve a risk-based STL control design problem, which we illustrate with numerical experiments.

\section{Risk Measures and Axiomatic Risk Theory} \label{sec:risk_measures}
A risk measure is a function that assigns a real number to a random variable, quantifying its size, typically related to one of its tails. More formally, let $(\Omega, \mathcal{F}, \mathbb{P})$ be a probability space, where $\Omega$ is the sample space, $\mathcal{F}$ is a $\sigma$-algebra of subsets of $\Omega$, and $\mathbb{P}$ is a probability measure on $\mathcal{F}$. Let $\mathbf{X}$ denote a set of real-valued random variables on $\Omega$ (i.e., $X\in \mathbf{X}$ is a Borel measurable function $X: \Omega \rightarrow \mathbb{R}$). A risk measure is a function $\rho : \mathbf{X} \rightarrow \mathbb{R} \cup \{\pm \infty\}$. In finance, elements of $\mathbf{X}$ represent the value of a financial position, and a risk measure quantifies the probability and severity of a financial loss. Here, elements of $\mathbf{X}$ will represent states of a stochastic system, and a \emph{risk measure quantifies the probability and severity of violating a signal temporal logic specification}.

\textbf{Risk Axioms.} Effective quantitative risk management in emerging complex autonomous systems is a major challenge, which motivates an axiomatic approach to risk measures. Four important axioms for a risk measure $\rho$ \cite{majumdar2020should} are:
\begin{enumerate}
    \item \textbf{Monotonicity:} If $X_1 \leq X_2$ then $\rho(X_1) \leq \rho(X_2)$;
    \item \textbf{Translation Invariance:} $\rho(X + c) = \rho(X) + c$;
    \item \textbf{Positive Homogeneity:} $\rho(\beta X) = \beta \rho(X)$;
    \item \textbf{Subadditivity:} $\rho(X_1 + X_2) \leq \rho(X_1) + \rho(X_2)$;
\end{enumerate}
$\forall X, X_1, X_2 \in \mathbf{X},$ and $c, \beta \in \mathbb{R}, \beta \geq 0$.
A risk measure is called \emph{coherent} if it satisfies all four of these axioms \cite{artzner1999coherent}. It has been argued that these axioms constitute natural desirable properties for risk measures in complex systems \cite{artzner1999coherent,majumdar2020should}. 
\emph{Spectral} or \emph{distortion risk measures} also satisfy
\begin{enumerate}
    \item[5)] \textbf{Law Invariance:} If $X_1, X_2 \in \mathbf{X}$ are identically distributed, then $\rho(X_1) = \rho(X_2)$.
    \item[6)] \textbf{Comonotone Additivity:} If $X_1, X_2 \in \mathbf{X}$ are comonotone (namely, $(X_1(\omega_2) - X_1(\omega_1))(X_2(\omega_2) - X_2(\omega_1)) \geq 0 \ \forall \omega_1, \omega_2 \in \Omega$), then $\rho(X_1 + X_2) = \rho(X_1) + \rho(X_2)$
\end{enumerate}
and can be viewed as refinements of coherent risk measures.

\textbf{Examples of Risk Measures.} Several commonly used and studied risk measures operating on $X \in \mathbf{X}$ include:
\begin{enumerate}
    \item \textbf{Expectation:} $\rho_E(X) = \mathbb{E}(X)$
    \item \textbf{Worst-Case:} $\rho_W(X) = \text{ess} \sup (X)$
    \item \textbf{Mean-Variance:} $MV_\lambda(X) = \mathbb{E}(X) + \lambda \text{Var}(X), \quad \lambda > 0$
    \item \textbf{Value at Risk (VaR) at level $\delta \in (0,1]$:} $\text{VaR}_{\delta}(X) = \inf_{t \in \mathbb{R}} \{ t \mid \mathbb{P}(X \leq t) \geq 1 - \delta \}$
    \item \textbf{Conditional Value at Risk (CVaR) at level $\delta \in (0,1]$:} $\text{CVaR}_{\delta}(X) = \inf_{t \in \mathbb{R}} \{ t + \frac{1}{\delta} \mathbb{E}[X - t]_+\}$,
    $[\cdot]_+ := \max(\cdot,0)$ 
    \item \textbf{Entropic Value at Risk (EVaR) at level $\delta \in (0,1]$:}  $\text{EVaR}_{\delta}(X) = \inf_{z > 0} \{ z^{-1} \ln( M_X(z)/\alpha ) \}$, 
    where $M_X(z) = \mathbb{E}(e^{zX})$ is the moment generating function of $X$. 
\end{enumerate}

Unfortunately, many widely used risk measures in robotics and engineering are not coherent, and can lead to serious miscalculations of risk, e.g., mean-variance and mean-standard-deviation fail to be monotone, and VaR lacks subadditivity. The widely used chance constraint in optimization models is closely related to VaR and has been used for notions of STL robustness for stochastic systems. We advocate for axiomatic risk theory with coherent risk as a more systematic and sophisticated approach to risk management for STL specifications in emerging safety-critical autonomous systems.

\textbf{Distributional Robustness.} Evaluating any of the above risk measures requires knowledge of the probability distribution of the associated random variable. In practice however, we are never given the probability distribution, only noisy data. Instead, we must estimate properties of the distribution from the noisy data, or make assumptions about the distribution. In the emerging area of distributionally robust optimization \cite{wiesemann2014distributionally}, this uncertainty in our knowledge of the probability distribution itself is explicitly accounted for. Rather than assuming a single probability distribution, we instead work with \emph{ambiguity sets} of distributions. Common examples of ambiguity sets include those with certain moment constraints, e.g., the set of distributions with given mean and variance $\mathscr{P} = \{\mathbb{P} \mid \mathbb{E}_\mathbb{P} [X] = \mu, \mathbb{E}_\mathbb{P} [(X-\mu)^2] = \Sigma \}$, or distance/divergence based sets, such as $\mathscr{P} = \{\mathbb{P} \mid d(\mathbb{P}, \hat{\mathbb{P}}) \leq \varepsilon \}$, where $d$ is a probability distance function, $\hat{\mathbb{P}}$ is a nominal distribution, and $\varepsilon > 0$ is a radius. A risk measure and an ambiguity set can be combined to obtain distributionally robust (DR) risk measures: $\rho_{\text{DR}} = \sup_{\mathbb{P} \in \mathscr{P}} \rho(X)$.
For example, DR-VaR, $\sup_{\mathbb{P} \in \mathscr{P}} \text{VaR}_\delta(X)$, with various moment-based ambiguity sets is a coherent risk measure \cite{zymler2013worst}.

\section {Risk-Based Signal Temporal Logic}
Traditionally, STL constraints are specified for \emph{deterministic} dynamical systems. 
STL formulas are based on predicates $\pi\in\{\top,\bot\}$ where $\top$ and $\bot$ denote true and false, respectively, and are obtained from evaluating a predicate function $\alpha: \mathbb{R}^n \rightarrow \mathbb{R}$ such that $\pi =\{ \alpha(x) \geq 0\}$, i.e., $\pi=\top$ if and only if $\alpha(x) \geq 0$ where $x\in\mathbb{R}$. A typical STL formula $\varphi$ is composed from logical and bounded-time temporal operators and can always be rewritten in \emph{negation normal form} \cite{fainekos2009robustness} (also referred to as positive normal form in\cite{sadraddini2015robust}), where the negation operator appears only at the atomic predicate level. 
We assume that STL formulas are given in negation normal form, which are defined recursively through the grammar:
\begin{align} \label{eq:stl_grammar}
    \varphi := \top | \pi | \neg \pi | \varphi^\prime \wedge \varphi^{\prime \prime} | \varphi^\prime \vee \varphi^{\prime \prime} | \varphi^\prime \mathscr{U}_{[a,b]} \varphi^{\prime \prime} | \varphi^\prime \mathscr{R}_{[a,b]} \varphi^{\prime \prime}
\end{align}
where $\neg$ is the negation operator, $\wedge$ and $\vee$ are conjunction and disjunction,
$\mathscr{U}_{[a,b]}$ and $\mathscr{R}_{[a,b]}$ are the until and release operators with $a \leq b$ and $a,b \in \mathbb{R}_{\geq 0} = [0, +\infty)$. Note that $\vee$ and $\mathscr{R}$ are duals of $\wedge$ and $\mathscr{U}$ respectively, since $\varphi^\prime \wedge \varphi^{\prime \prime} = \neg (\neg \varphi^\prime \vee \neg \varphi^{\prime \prime})$ and $\varphi^\prime \mathscr{U}_{[a,b]} \varphi^{\prime \prime} = \neg(\neg \varphi^\prime \mathscr{R}_{[a,b]} \neg \varphi^{\prime \prime})$. The negation normal form replaces the negation of a formula by including all operators and their duals in the grammar. 
We can also define additional operators such as
eventually ($\diamondsuit_{[a,b]} \varphi := \top \mathscr{U}_{[a,b]} \varphi$)
and always ($\square_{[a,b]} \varphi:= \bot \mathscr{R}_{[a,b]} \varphi$).

Denote by $\xi_N$ a \emph{deterministic} discrete-time finite-horizon system run of length $N+1$. If the signal starts at time $t$, then $ (\xi_N,t) = x_t x_{t+1} \dots x_{t+N}$ where $x_i$ is the \emph{deterministic} state at time $i$. $(\xi_N,t) \models \varphi$ denotes $\xi_N$ satisfying $\varphi$ (starting) at time $t$. By convention, $\xi_N$ satisfies $\varphi$ (denoted $\xi_N \models \varphi$) if $(\xi_N, 0)\models \varphi$.
The traditional STL semantics are defined as: 
$(\xi_N, t) \models \top$ trivially holds, 
$(\xi_N, t) \models \pi$ iff $\alpha(x_t) \geq 0$, 
$(\xi_N, t) \models \neg \pi$ iff $\neg ((\xi_N, t) \models \pi)$, 
$(\xi_N, t) \models \varphi \wedge \varphi^\prime$ iff $(\xi_N, t) \models \varphi \wedge (\xi_N, t) \models \varphi^\prime$, 
$(\xi_N, t) \models \varphi \vee \varphi^\prime$ iff $(\xi_N, t) \models \varphi \vee (\xi_N, t) \models \varphi^\prime$, 
$(\xi_N, t) \models \varphi \mathscr{U}_{[a,b]} \varphi^\prime$ iff $\exists t^\prime \in [t+a, t+b]$ s.t. $(\xi_N, t^\prime) \models \varphi^\prime \wedge \forall t^{\prime \prime} \in [t, t^\prime], (\xi_N, t^{\prime \prime}) \models \varphi$, 
and $(\xi_N, t) \models \varphi \mathscr{R}_{[a,b]} \varphi^\prime$ iff $\forall t^\prime \in [t+a, t+b]$, $(\xi_N, t^\prime) \models \varphi^\prime \nonumber \vee \exists t^{\prime \prime} \in [t, t^\prime], (\xi_N, t^{\prime \prime}) \models  \varphi$. 
An STL formula horizon can be upper bounded by the maximum over the sums of the nested upper bounds of the temporal operators. The bound is denoted $\text{len}(\varphi)$ and defined recursively as:
$\text{len}(\top) = \text{len}(\pi) = \text{len}(\neg \pi) = 0$,
$\text{len}(\varphi \wedge \varphi^\prime) =\text{len}(\varphi \vee \varphi^\prime)= \max(\text{len}(\varphi), \text{len}(\varphi^\prime))$, 
and $\text{len}(\varphi \mathscr{U}_{[a,b]} \varphi^\prime) =\text{len}(\varphi \mathscr{R}_{[a,b]} \varphi^\prime)= b + \max(\text{len}(\varphi),\text{len}(\varphi^\prime))$. 
As opposed to the Boolean nature of the above semantics, STL allows to define a quantitative robustness measure, which is a function $\rho^\varphi(\xi_N,t)$ \cite{donze2010robust} that assigns the formula $\varphi$ a real value such that $(\xi_N,t) \models \varphi$ if $\rho^\varphi (\xi_N,t) > 0$, indicating how robustly the formula is satisfied. Note that $\rho^\varphi(\xi_N,t)$ is deterministic.

\emph{Stochastic} systems require an alternative quantitative semantics. We denote a discrete-time finite-horizon \emph{stochastic process} starting at time $t$ as ($\Xi_N, t) = X_t X_{t+1} \dots X_{t+N}$ where $X_i$ is an $\mathbb{R}^n$-valued random variable from a set $\mathbf{X}$ of random variables that represents the \emph{stochastic} system state at time $i$. Since the system state is stochastic, the question of $\Xi_N$ satisfying an STL formula is ill-posed, and the aforementioned Boolean semantics do not apply (notice that $\alpha(X)$ is a random variable). 
Instead, we quantify the risk of \emph{violating} a specification by introducing a risk measure into STL formulas. (This contrasts with STL robustness measures, which quantify the \emph{satisfaction} of a formula; it is more natural to consider the risk of violating a formula.)
While there are several ways to define the risk of \emph{violating} a formula $\mathcal{R}$,
we choose to define the risk of violating an atomic predicate using the risk measures $\rho$ discussed in Section \ref{sec:risk_measures} and build up the STL risk-of-violation semantics recursively.

\begin{definition} {\emph{STL Risk Semantics.}} \label{def:stl_risk_semantics}
    For the stochastic system run $\Xi_N$, risk measure $\rho$, atomic predicate $\pi$, and STL formulas $\varphi$ and $\varphi^\prime$, the STL risk is defined recursively as
    \begin{itemize}
        \item $\mathcal{R} ((\Xi_N, t) \not\models \top) := -\infty$, 
        \item $\mathcal{R} ((\Xi_N, t) \not\models \bot) := +\infty$, 
        \item $\mathcal{R} ((\Xi_N, t) \not\models \pi) := \rho(-\alpha(X_t))$,  
        \item $\mathcal{R} ((\Xi_N, t) \not\models \neg\pi) := \rho(\alpha(X_t))$,
        \item $\mathcal{R}((\Xi_N, t) \not\models \varphi \wedge \varphi^\prime) := \max (\mathcal{R} ((\Xi_N, t) \not\models \varphi),\\ 
        \mathcal{R}((\Xi_N, t) \not\models \varphi^\prime))$,
        \item $\mathcal{R}((\Xi_N, t) \not\models \varphi \lor \varphi^\prime) := \min (\mathcal{R} ((\Xi_N, t) \not\models \varphi),\\
        \mathcal{R}((\Xi_N, t) \not\models \varphi^\prime))$,
        \item $\mathcal{R} ((\Xi_N,t) \not\models \varphi \mathscr{U}_{[a, b]} \varphi^\prime) := \min_{i \in [a,b]}( \max(\\ 
        \mathcal{R} ((\Xi_N,t+i) \not\models \varphi^\prime),\max_{j \in [0,i]} \mathcal{R} ((\Xi_N, t+j) \not\models \varphi)))$,
        \item $\mathcal{R} ((\Xi_N, t) \not\models \varphi \mathscr{R}_{[a, b]} \varphi^\prime) := \max_{i \in [a,b]}(  \min(\\
        \mathcal{R} ((\Xi_N, t+i) \not\models \varphi^\prime), \min_{j \in [0,i]} \mathcal{R} ((\Xi_N, t+j) \not\models \varphi)))$.  
    \end{itemize}
\end{definition}
Often it is desirable to bound the STL risk of violation of a formula $\varphi$ by some prescribed value $\delta$, i.e., $\mathcal{R}((\Xi_N, t) \not \models \varphi) \leq \delta$.
The risk metrics of Section \ref{sec:risk_measures} can be used to quantify the aforementioned risk. It must be noted that these metrics are generally asymmetric, i.e. $\rho(- \alpha(X_t)) \not= -\rho(\alpha(X_t))$, and appear only at the atomic level
(which further justifies using the grammar \eqref{eq:stl_grammar}).
This will be explored in Section \ref{sec:reform_stl_constraint}.

\section{Control Design for Risk-Constrained STL} \label{sec:ctrl_design}
We consider the discrete-time stochastic linear system:
\begin{equation} \label{eq:dts_lin_sys}
X_{t+1} = A_t X_t + B_t u_t + W_t, \ \  X_0 = x_0 ,
\end{equation}
where $X_t \in \mathbb{R}^n$, $u_t \in \mathbb{R}^m$, and $W_t \in \mathbb{R}^n$ are respectively the state, control input, and stochastic disturbance of the system at time $t$,
$A_t \in \mathbb{R}^{n \times n}$ and $B_t \in \mathbb{R}^{n \times m}$ are respectively the system dynamics and input matrices, and $x_0$ is the known initial system state. 
The disturbances $\{W_t\}$ are assumed independent according to unknown distributions $\mathbb{P}_{W_t}$, which belong to ambiguity sets $\mathscr{P}^W_t$. Here we consider moment-based ambiguity sets with given mean and variance:
$\mathscr{P}^{W}_t = \{ \mathbb{P}_{W_t} \ |\  \mathbb{E}[W_t] = \overline{W}_t, \ \mathbb{E}[(W_t - \overline{W}_t)(W_t - \overline{W}_t)^T] = \Sigma_{W_t}\}.$

From \eqref{eq:dts_lin_sys} we obtain
\begin{align} \label{eq:lin_sys_soln}
    X_\tau = \Phi(\tau, t)X_t + \Sigma_{k = 1}^{\tau-1}\Phi(\tau, k+1)(B_k u_k+W_k),
\end{align}
where $\Phi(.,.)$ is the state transition matrix defined as
\begin{align} \label{eq:state_trans_mat}
    \Phi(\tau,t) = 
    \begin{dcases}
        A_{\tau-1}A_{\tau-2}\dots A_t \ \ \ &\tau > t \geq 0 \\
        I_n    &\tau = t \geq 0.
    \end{dcases}
\end{align}

We consider a finite-horizon control design problem where the goal is to determine a state feedback control policy for $\eqref{eq:dts_lin_sys}$ that satisfies a risk constraint associated with an STL formula $\varphi$. Specifically, we consider
\begin{align} \label{socrisk}
    J^*(x_0)= & \min_{\mu \in \mathcal{M}} \ \mathbb{E} [J({X}_{0: N}, {u}_{0:N-1})]\ \ \ \text{subject to}  \\  
    &X_t = \Phi(t, 0)x_0 + \Sigma_{k=0}^{t-1}\Phi(t,k+1)(B_k u_k + W_k),  \nonumber \\ 
    &\mathcal{R}((\Xi_N,0) \not\models \varphi) \leq \delta,  \quad
    {u}_{0:N-1} \in U \nonumber,
\end{align}
where the decision variable is a state feedback policy $\mu = [\mu_0, \mu_1, \dots, \mu_{N-1}]$ with $u_t=\mu_t(x_t)$ from a set $\mathcal{M}$ of measurable policies, ${u}_{t:N-1} = [u^T_t, u^T_{t+1}, \dots, u^T_{N-1}]^T$, $X_{t:N} = [X^T_t, X^T_{t+1}, \dots, X^T_{N+1}]^T$ and ${W}_{t:N-1} = [W^T_t, W^T_{t+1}, \dots, W^T_{N-1}]^T$, $N$ is the time horizon, $J$ is a stage cost function with expectation taken with respect to the disturbance sequence $\{W_t\}$, $\delta \in \mathbb{R}$ is a user-defined risk bound, and $U \subset \mathbb{R}^{Nm}$ is an input constraint set. The challenge lies in the uncertainty in the system model and its appearance in the STL risk constraint. We approach this using Model Predictive Control (MPC) and a reformulation of the STL risk constraint into a tightened deterministic STL constraint so that existing control approaches for deterministic STL constraints can be used such as summarized in Section~\ref{sec:int}.

Using MPC to (approximately) solve \eqref{socrisk} reduces the problem to a sequence of open-loop optimization problems.
For an STL problem with formula $\varphi$, a natural choice for the prediction horizon is $N \geq \text{len}(\varphi)$ with a system run of $N_s>N$. 
The MPC optimization problem at time $t \leq N_s - N$ given the observed state, $x_t$, is:
\begin{align} \label{eq:mpc_prob}
    J^*(x_t)=&\min_{u_{t:t+N-1}} \ \mathbb{E} [J(X_{t:t+N}, u_{t:t+N-1})]\ \ \ \text{subject to}  \\  
    &X_\tau = \Phi(\tau, t)x_t + \Sigma_{k=t}^{\tau-1}\Phi(\tau,k+1)(B_k u_k + W_k ),  \nonumber \\ 
    &\mathcal{R}((\Xi_{N}, 0) \not\models \varphi) \leq \delta, \quad u_{t:t+N-1} \in U, \nonumber
\end{align}
where the decision variable is the open-loop control sequence $u_{t:t+N-1}$, expectation is with respect to $W_{t:t+N-1}$. Solving the optimization problem yields the future optimal control sequence $u^*_{t:t+N-1}=[u^*_t, \dots, u^*_{t+N-1}]$. Only the first component $u^*_t$ of this plan is implemented, and the problem is solved again after the next state realization is observed. Thus, the MPC policy is $\mu_{\text{MPC}}(x_t) = u^*_t$. 
We assume that $\mathbb{E} [J(X_{t:t+N}, u_{t:t+N-1})]$ in \eqref{eq:mpc_prob} can be evaluated analytically, which is the case for quadratic $J$, allowing us to focus solely on challenges in accounting for the STL risk constraint.

\section{Reformulation of STL Risk Constraints} \label{sec:reform_stl_constraint}
In this section, we demonstrate how the optimization problem \eqref{eq:mpc_prob} can be reformulated into an optimization problem with \emph{deterministic}, tightened STL constraints on the nominal system dynamics. We also show how to explicitly write the tightened constraints on atomic predicates for one specific coherent risk measure, Distributionally Robust Value at Risk.

\subsection{From STL formula violation risk to atomic predicate violation risk} \label{sec:formula_risk_to_atomic_risk}
The semantics in Definition \ref{def:stl_risk_semantics} are useful for two main reasons: 1) the risk metrics described in Section \ref{sec:risk_measures} appear only at the atomic predicate level and the risk of failing to satisfy an STL formula is defined recursively from there, 2) through the recursive STL risk semantics, all operators (except negation) are defined in terms of $\min$ and $\max$ operators over time intervals. These reasons allow transforming a risk-based STL constraint into similar risk constraints on atomic predicates. This is formalized in Theorem \ref{thm:stl_risk_to_predicates}.

\begin{theorem}{}\label{thm:stl_risk_to_predicates}
Consider a risk measure $\rho$ from Section \ref{sec:risk_measures}, STL risk semantics $\mathcal{R}$ from Definition \ref{def:stl_risk_semantics}, STL formula $\varphi$ given by \eqref{eq:stl_grammar}, and risk bound $\delta \in \mathbb{R}$.
The STL risk constraint $\mathcal{R}((\Xi_N, t)\not\models \varphi) \leq \delta$ can be transformed into conjunctions and disjunctions over time intervals of risk constraints on atomic predicates of the form $\rho(\pm \alpha(X_t)) \leq \delta$ where $\alpha$ is a predicate function associated with the atomic predicate $\pi$.
\end{theorem}

\begin{proof}
The proof is a consequence of the grammar \eqref{eq:stl_grammar} and the structure of the risk semantics which apply the risk measure only on the atomic predicate level. We consider eight cases (other cases follow recursively):  
    \begin{enumerate}
        \item $\mathcal{R} ((\Xi_N, t) \not\models \top) \leq \delta \Leftrightarrow -\infty \leq \delta$ \quad \text{(always true)}
        \item $\mathcal{R} ((\Xi_N, t) \not\models \bot) \leq \delta \Leftrightarrow +\infty \leq \delta$ \quad \text{(always false)}
        \item $\mathcal{R} ((\Xi_N, t) \not\models \pi) \leq \delta \Leftrightarrow \rho(-\alpha(X_t)) \leq \delta$
        \item $\mathcal{R} ((\Xi_N, t) \not\models \neg\pi) \leq \delta \Leftrightarrow \rho(\alpha(X_t))\leq \delta$
        \item $\mathcal{R}((\Xi_N, t) \not\models \varphi \wedge \varphi^\prime) \leq \delta \\
        \Leftrightarrow \max (\mathcal{R} ((\Xi_N, t) \not\models \varphi), \mathcal{R}((\Xi_N, t) \not\models \varphi^\prime))\leq \delta \\
        \Leftrightarrow \mathcal{R}((\Xi_N, t) \not\models \varphi)\leq \delta \wedge \mathcal{R}((\Xi_N, t) \not\models \varphi^\prime)\leq \delta$
        \item $\mathcal{R}((\Xi_N, t) \not\models \varphi \vee \varphi^\prime) \leq \delta \\
        \Leftrightarrow \min (\mathcal{R} ((\Xi_N, t) \not\models \varphi), \mathcal{R}((\Xi_N, t) \not\models \varphi^\prime))\leq \delta \\
        \Leftrightarrow \mathcal{R}((\Xi_N, t) \not\models \varphi)\leq \delta \nonumber \vee \mathcal{R}((\Xi_N, t) \not\models \varphi^\prime)\leq \delta $
        \item $\mathcal{R} ((\Xi_N, t) \not\models \varphi \mathscr{U}_{[a, b]} \varphi^\prime) \leq \delta \\ \Leftrightarrow \min_{i \in [a,b]}( \max(\mathcal{R} ((\Xi_N, t+i) \not\models \varphi^\prime), \\
        \max_{j \in [0,i]} \mathcal{R} ((\Xi_N, t+j) \not\models \varphi))) \leq \delta \\
        \Leftrightarrow \exists t^\prime \in [t+a, t+b] \mathcal{R}((\Xi_N,t^\prime)\not\models \varphi^\prime) \leq \delta \\ \wedge \forall t^{\prime \prime} \in [t, t^\prime] \mathcal{R}((\Xi_N,t^\prime)\not\models \varphi) \leq \delta$
        \item $\mathcal{R} ((\Xi_N, t) \not\models \varphi \mathscr{R}_{[a, b]} \varphi^\prime) \leq \delta \\ \Leftrightarrow \max_{i \in [a,b]}(\min(\mathcal{R} ((\Xi_N, t+i) \not\models \varphi^\prime), \\ 
        \min_{j \in [0,i]} \mathcal{R} ((\Xi_N, t+j) \not\models \varphi))) \leq \delta \\
        \Leftrightarrow \forall t^\prime \in [t+a, t+b] \mathcal{R}((\Xi_N,t^\prime)\not\models \varphi^\prime) \leq \delta \\
        \vee \exists t^{\prime \prime} \in [t, t^\prime] \mathcal{R}((\Xi_N,t^\prime)\not\models \varphi) \leq \delta$
\end{enumerate}
The $\forall$ and $\exists$ operators can be expressed as chains of $\wedge$ and $\vee$ operators respectively across different time steps resulting in time intervals over which conjunctions and disjunctions of deterministic constraints (lines 3 and 4) are formed.
Lines 1 and 2 are trivial and can be safely discarded.
\end{proof}

Using Theorem \ref{thm:stl_risk_to_predicates}, it is clear that 
the risk level $\delta$ is carried through the formula and only appears at the atomic level. All operators, other than negations, are also preserved and they reappear to connect the risk of violating atomic predicates; e.g. the risk of failing to satisfy the conjunction of two predicates is the conjunction of the risk of failing to satisfy each predicate. Furthermore, the risk measure of an atomic predicate integrates the stochasticity in the state, leaving us with \emph{deterministic} inequalities of the form:
\begin{align}
    & \rho(-\alpha(X_t)) \leq \delta & 
    & \Leftrightarrow &
    & \delta - \rho(-\alpha(X_t)) \geq 0 & \label{eq:delta_minus_rho_minus}\\
    & \rho(\alpha(X_t)) \leq \delta & 
    & \Leftrightarrow &
    & \delta - \rho(\alpha(X_t)) \geq 0. & \label{eq:delta_minus_rho}
\end{align}

\subsection{Risk-tightened predicates}\label{sec:risk_tightened_predicates}
Having found atomic predicate risk constraints \eqref{eq:delta_minus_rho_minus}  and \eqref{eq:delta_minus_rho}, we now turn to reformulating the stochastic STL problem into a deterministic one assuming affine predicates\footnote{Note that the use of affine predicates only is not particularly restrictive since most Mixed-Integer Linear Programming (MILP) tools, widely used with STL, only allow affine predicates.}.
Consider the following rearrangement of \eqref{eq:dts_lin_sys}:
\vspace{-0.2cm}
\begin{align}\label{eq:X_reform}
    X_t = & \overbrace{A_{t-1} A_{t-2} \dots A_0}^{:=G_{t-1}} x_0 \\ 
    & + \overbrace{\begin{bmatrix}
        A_{t-1} \dots A_1 B_0 
        & \dots 
        & A_{t-1} B_{t-2} 
        & B_{t-1} 
    \end{bmatrix}}^{:=H_{t-1}}
    \begin{bmatrix}
        u_0^T \dots  u_{t-1}^T
    \end{bmatrix}^T
     \nonumber \\
    & + \overbrace{\begin{bmatrix}
        A_{t-1} \dots A_1 
        & \dots
        & A_{t-1}
        & I
    \end{bmatrix}}^{:=L_{t-1}}
    \begin{bmatrix}
        W_0^T \dots  W_{t-1}^T
    \end{bmatrix}^T \nonumber \\ 
    = & G_{t-1}x_0 + H_{t-1}u_{0:t-1} + L_{t-1}W_{0:t-1}. \nonumber
\end{align}
    
Any random variable can be written as a sum of its expectation and a zero-mean random variable. Let $\overline{W}_{0:t-1} := \mathbb{E}[W_{0:t-1}]$ ($\overline{W}_{t} := \mathbb{E}[W_{t}]$). We thus have:
\begin{align}
    X_t &= \overline{x}_t + \hat{X}_t \label{eq:X_sum}\\
    \overline{x}_t &= \mathbb{E}[X_t] = G_{t-1}x_0 + H_{t-1}u_{0:t-1} + L_{t-1}\overline{W}_{0:t-1} \label{eq:X_mean}\\
    \hat{X}_t &= L_{t-1}(W_{0:t-1} - \overline{W}_{0:t-1}) \label{eq:X_0mean}
\end{align}
It is easy to see that \eqref{eq:X_mean} results in the deterministic system: 
\begin{align}
    \overline{x}_t &= A_{t-1}\overline{x}_{t-1} + B_{t-1}\overline{u}_{t-1} + \overline{W}_{t-1}. \label{eq:x_det_sys}
\end{align}

Using \eqref{eq:X_sum} in \eqref{eq:delta_minus_rho_minus} and \eqref{eq:delta_minus_rho}, assuming affine predicates $\alpha(X_t) = a^T X_t + b$, and using \emph{translation invariance} of $\rho$ we get:
\vspace*{-.5cm}
\begin{align} \label{eq:det_affine_pred}
    \delta - \rho(\pm \alpha(X_t)) \geq 0 
    &\Leftrightarrow 
    \overbrace{\delta - \rho(\pm a^T \hat{X}_t)}^{:=\overline{\delta}_{\rho\pm}(\delta,\hat{X}_t)} \mp \alpha(\overline{x}_t) \geq 0
\end{align}
    
Notice that $\overline{\delta}_{\rho*}: \mathbb{R} \times \mathbf{X} \rightarrow \mathbb{R} \cup \{\pm \infty\}$ is deterministic ($*\in\{+,-\}$). Thus, \eqref{eq:det_affine_pred} is affine in $\overline{x}_t$ and can be used as an \emph{affine deterministic predicate function}.
\begin{definition}{\emph{Risk-Tightened Predicates}.} \label{def:risk_tightened_predicate}
    Given risk metric $\rho$, risk bound $\delta \in \mathbb{R}$, affine predicate function $\alpha$ and random variable $X_t$ per \eqref{eq:X_sum}, we define risk-tightened affine predicates as: 
    \begin{align}
        \overline{\pi}&:= 
        \{\overline{\alpha}_{\rho*}(\delta, X_t)\geq 0\} \nonumber \\
        &:=
        \begin{dcases}
            \{\overline{\delta}_{\rho-}(\delta,\hat{X}_t) + \alpha(\overline{x}_t) \geq 0\} \quad \text{if } *=- \\
            \{\overline{\delta}_{\rho+}(\delta,\hat{X}_t) - \alpha(\overline{x}_t) \geq 0 \} \quad \text{if } *=+
        \end{dcases}
    \end{align}
\end{definition} 

We now present a corollary that establishes the form of the deterministic STL constraint and the corresponding deterministic system.
\begin{corollary}\label{cor:det_constraint_to_det_formula}
    Consider STL formula $\varphi$ given by \eqref{eq:stl_grammar}. The risk constraint $\mathcal{R}((\Xi_N,t)\not\models \varphi) \leq \delta$ on \eqref{eq:dts_lin_sys} with affine predicates is equivalent to a \emph{deterministic} STL constraint $\overline{\varphi}$ on the mean trajectory $\overline{\xi}_N = \overline{x}_t \dots \overline{x}_{t+N}$ of \eqref{eq:dts_lin_sys} given by \eqref{eq:x_det_sys} with affine predicates per Definition \ref{def:risk_tightened_predicate}.
\end{corollary}

\begin{proof}
    The proof follows from the reformulation of \eqref{eq:dts_lin_sys} as \eqref{eq:X_reform} and using \eqref{eq:X_sum} in the affine predicate $\alpha(X_t)$. This results in the deterministic system \eqref{eq:x_det_sys} and $\rho(\pm\alpha(X_t))\leq \delta$ iff $\overline{\alpha}_{\rho \pm}(\delta, X_t) \geq 0$.
    $\overline{\varphi}$ follows from Definitions \ref{def:stl_risk_semantics}  and \ref{def:risk_tightened_predicate}.
\end{proof}

\vspace*{-.35cm}
\begin{example} \label{ex:new_phi_derivation}
    Consider the atomic predicates $\pi_1, \pi_2, \pi_3$ corresponding to $\{\alpha_1(x) \geq 0\}, \{\alpha_2(x) \geq 0\}, \{\alpha_3(x) \geq 0\}$ respectively (affine predicate functions) and the formula $\varphi = \pi_1 \mathscr{U}_{[3,5]} (\pi_2 \wedge \neg \pi_3)$. 
    The risk of violating $\varphi$ with $\delta$ risk bound is $\mathcal{R}((\Xi_5, 0) \not \models \varphi) \leq \delta$.
     Theorem \ref{thm:stl_risk_to_predicates} and Definition \ref{def:stl_risk_semantics} yield\vspace*{-.25cm}
        \begin{multline*}
            \mathcal{R}((\Xi_5, 0) \not \models \pi_1 
            \mathscr{U}_{[3,5]} (\pi_2 \wedge \neg \pi_3)) 
            \leq \delta  \nonumber \\
            \Leftrightarrow  
            \exists t^\prime \in [3,5] \mathcal{R}((\Xi_5, t^\prime)\not \models \pi_2) \leq \delta \wedge 
            \mathcal{R}((\Xi_5, t^\prime)\not \models \neg \pi_3) \leq \delta  \nonumber \\
            \shoveright{\wedge \forall t^{\prime \prime} \in [0, t^\prime] \mathcal{R}((\Xi_5, t^{\prime \prime}) \not \models \pi_1) \leq \delta} \nonumber \\   
            {\Leftrightarrow  \exists t^\prime \in [3,5] \rho(-\alpha_2(X_{t^\prime}))) \leq \delta \wedge \rho(\alpha_3(X_{t^\prime})) \leq \delta} \\
            \wedge \forall t^{\prime \prime} \in [0, t^\prime] \rho(-\alpha_1(X_{t^{\prime \prime}})) \leq \delta. \nonumber
        \end{multline*}
        Using Definition \ref{def:risk_tightened_predicate} and letting 
        $\overline{\pi}_1 = \{\overline{\alpha}_{1\rho-} \geq 0\}$, $\overline{\pi}_2 = \{\overline{\alpha}_{2\rho-} \geq 0\}$, and $\overline{\pi}_3 = \{\overline{\alpha}_{3\rho+} \geq 0\}$
        result in\vspace*{-.25cm}
        \begin{multline*}
            \mathcal{R}((\Xi_5, 0) \not \models \pi_1 \mathscr{U}_{[3,5]} (\pi_2 \wedge \neg \pi_3)) \leq \delta  \\
            \Leftrightarrow \exists t^\prime\in [3,5] \overline{\alpha}_{2\rho-}(\delta, X_{t^\prime})\geq 0 \wedge \overline{\alpha}_{3\rho+}(\delta, X_{t^\prime})\geq 0 \nonumber \\
            \wedge \forall t^{\prime \prime} \in [0, t^\prime] \overline{\alpha}_{1\rho-}(\delta, X_{t^{\prime \prime}})\geq 0 \\
            \Leftrightarrow  \exists t^\prime \in [3,5] (\overline{\xi}_5, t^\prime)\models \overline{\pi}_2 \wedge (\overline{\xi}_5, t^\prime)\models \overline{\pi}_3 \\
            \wedge \forall t^{\prime \prime} \in [0, t^\prime] (\overline{\xi}_5, t^{\prime \prime})\models \overline{\pi}_1 \\
            \Leftrightarrow (\overline{\xi}_5, 0) \models \overline{\pi}_1\mathscr{U}_{[3,5]}(\overline{\pi}_2 \wedge \overline{\pi}_3)
            \Rightarrow \overline{\varphi} := \overline{\pi}_1\mathscr{U}_{[3,5]}(\overline{\pi}_2 \wedge \overline{\pi}_3).
    \end{multline*}
\end{example}

\subsection{Explicit reformulation for affine predicates and Distributionally Robust Value at Risk} \label{sec:reformulation_affine_drvar}
In Section \ref{sec:formula_risk_to_atomic_risk} we obtained constraints on atomic predicates of the form: $\rho(-\alpha(X_t)) \leq \delta$ or $\rho(\alpha(X_t)) \leq \delta$. In Section \ref{sec:risk_tightened_predicates} we reformulated the system, derived risk-tightened atomic predicate constraints with affine predicates ($\alpha(x) = a^T x+b$), and presented the deterministic system. In this section we turn to evaluating the tightened constraints for a particular choice of the risk metric $\rho$: the Distributionally Robust Value at Risk (DR-VaR).
Under DR-VaR with ambiguity set $\mathscr{P} = \mathscr{P}^W_t$ from Section \ref{sec:ctrl_design}, the atomic predicate risk constraint, with $\delta$ constrained to $(0,1)$, can be rewritten as:
\begin{align}
    \rho(\alpha(X_t)) \leq \delta &\Leftrightarrow \sup_{\mathbb{P} \in \mathscr{P}}\mathbb{P}[a^T X_t +b \geq 0] \leq \delta \nonumber \\
    &\Leftrightarrow \inf_{\mathbb{P} \in \mathscr{P}}\mathbb{P}[a^T X_t +b \leq 0] \geq 1-\delta  \label{eq:rho_plus_as_inf} \\
    \rho(-\alpha(X_t)) \leq \delta &\Leftrightarrow \sup_{\mathbb{P} \in \mathscr{P}}\mathbb{P}[-a^T X_t -b \geq 0] \leq \delta \nonumber \\
    &\Leftrightarrow \inf_{\mathbb{P} \in \mathscr{P}}\mathbb{P}[-a^T X_t -b \leq 0] \geq 1-\delta   \label{eq:rho_minus_as_inf}
\end{align} 

We now present Lemma \ref{lemma:deterministic_constraints} which explicitly presents the deterministic constraints equivalent to \eqref{eq:rho_plus_as_inf} and \eqref{eq:rho_minus_as_inf}. 

\begin{lemma} \label{lemma:deterministic_constraints}
    The DR-VaR constraints \eqref{eq:rho_plus_as_inf} and \eqref{eq:rho_minus_as_inf} are equivalent to the deterministic affine inequality constraints
    \begin{align}
        a^T \overline{x}_t + b + \Delta \| \Sigma_{W_{0:t-1}}^{1/2}L_{t-1}^T a \|_2 \leq 0, \label{eq:rho_plus_deterministic_constraint} \\
        a^T \overline{x}_t + b - \Delta \| \Sigma_{W_{0:t-1}}^{1/2}L_{t-1}^T a \|_2 \geq 0 \label{eq:rho_minus_deterministic_constraint}
   \end{align}
    respectively, where $\Delta = \sqrt{(1-\delta)/\delta}$, $\Sigma_{W_{0:t-1}} = \mathbb{E}[(W_{0:t-1}-\overline{W}_{0:t-1})(W_{0:t-1}-\overline{W}_{0:t-1})^T]$, $\overline{W}_{0:t-1} = \mathbb{E}[W_{0:t-1}]$, $\overline{x}_t$ is the deterministic state 
    given by \eqref{eq:x_det_sys}
    with $\overline{x}_0 = x_0$.
\end{lemma}

\begin{proof}
    The proof follows by using \eqref{eq:x_det_sys} in \eqref{eq:rho_plus_as_inf} and \eqref{eq:rho_minus_as_inf} and applying \cite[Theorem 3.1]{calafiore2006distributionally} to both.
\end{proof}

\section{Numerical Experiments}
Consider $X:=(X_{x1},X_{y1}, X_{x2}, X_{y2})$ to be a four dimensional random variable and the following formula:
\begin{align} \label{eq:ex_formula}
    \varphi:=\square_{[0,3]}(\varphi_{a1} \wedge \varphi_{a2}) \wedge \diamondsuit_{[2,3]} \varphi_r \wedge \square_{[1,3]}\varphi_f
\end{align} 
where $\varphi_{a1}:= \pi_1\vee\pi_2\vee \pi_3\vee \pi_4$, $\varphi_{a2}:= \pi_5\vee\pi_6\vee \pi_7\vee \pi_8$, $\varphi_r:=\pi_9\wedge\pi_{10}\wedge \pi_{11}\wedge \pi_{12}$, and $\varphi_f:=\pi_{13} \wedge \pi_{14} \wedge \pi_{15} \wedge \pi_{16}$ encode two agents $ag_1 := (X_{x1},X_{y1})$ and $ag_2 := (X_{x2},X_{y2})$ avoiding a square set of side length 1 centered at $P_a:=(0,0)$ and agent $ag_1$ reaching a square set of side length 1 centered at $P_r := (2,0)$ while the two retain a maximum infinity norm distance $d:=1$ of each other. 
We have:
$\alpha_1(X) :=-X_{x1}-0.5$,
$\alpha_2(X) := X_{x1}-0.5$,
$\alpha_3(X):=-X_{y1}-0.5$,
$\alpha_4(X):= X_{y1}-0.5$,
$\alpha_5-\alpha_8$ are similar,
$\alpha_9(X):= X_{x1}-1.5$,
$\alpha_{10}(X):= 2.5-X_{x1}$,
$\alpha_{11}(X):= 0.5+X_{y1}$,
$\alpha_{12}(X):= 0.5-X_{y1}$,
$\alpha_{13}(X):= 1-X_{x1}+X_{x2}$,
$\alpha_{14}(X):= 1+X_{x1}-X_{x2}$,
$\alpha_{15}(X):= 1-X_{y1}+X_{y2}$,
$\alpha_{16}(X):= 1+X_{y1}-X_{y2}$.

We use the BluSTL \cite{donze2015blustl} package in its closed-loop deterministic setting for our results. 
The package uses YALMIP \cite{lofberg2004yalmip} and reformulates an STL specification into a Mixed-Integer Linear Program (MILP), and then solves an MPC problem.
Each agent is a double integrator with four states (position and velocity along two directions), two control inputs (along velocity directions), and four additive disturbances on the states. 
The disturbances are sampled from a 0-mean 3 degree of freedom t-distribution scaled to 0.005-variance and applied to velocity states only.
BluSTL converts the continuous time system to discrete-time with 0.1sec system steps and 0.2sec controller steps. 
We perform two types of simulations. 
Type 1: The controller ignores the disturbance in the dynamics and uses deterministic MPC with the nominal model, but is evaluated with the disturbance in closed-loop.
Type 2: We explicitly incorporate the uncertainty using our proposed STL risk analysis, transform the specification into a deterministic risk-tightened STL formula, compute the control, and then evaluate in closed-loop with the disturbance. We use $\delta_1 = 0.1$ risk bound for predicates involving the obstacle avoidance and $\delta_2 = 0.5$ risk bound for predicates involving the goal region. This puts more emphasis on obstacle avoidance. We also saturate the tightening parameter $\Delta \|\Sigma_{W_{0:t-1}}^{1/2} L_{t-1}^T a \|_2$ after 1 second to limit the effect of increasing prediction uncertainty.
In both cases the objective $J$ minimizes the sum of the 1-norm of the control inputs and the distance to the goal. 
We run 100 simulations for each type and present them in Fig.~\ref{fig:results}.

\begin{figure}[htb]
\begin{center}
    \includegraphics[width=1.0\columnwidth]{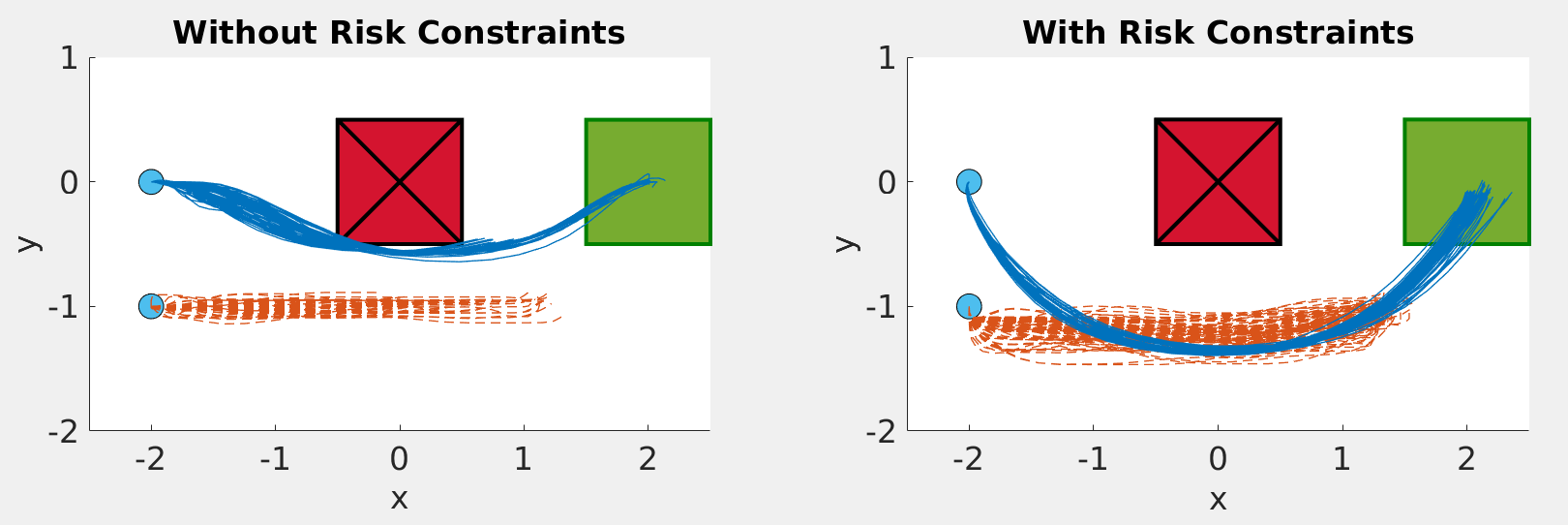}
    \caption{Type 1 (left) and type 2 (right) simulations of 100 sample trajectories for the task \eqref{eq:ex_formula}. Agent 1 (solid blue) and agent 2 (dashed orange) start from the center of the two blue circles, must avoid the obstacle (red square with cross) and agent 1 must reach the goal (green square).}
    \label{fig:results}
\end{center}
\end{figure}

\vspace{-0.2cm}
Type 1 trajectories (left) get very close and occasionally collide with the obstacle. The disturbance causes 84 of the 100 simulations to fail to reach the goal. 
Type 2 (right) trajectories, however, satisfy the risk bounds and remain sufficiently far from the obstacle, and 98 of the 100 reach the goal.
\vspace{-0.05cm}

\section{Conclusion and Future Work}
We presented a general framework for risk-based STL specifications for stochastic systems using axiomatic risk theory. 
We are exploring several extensions and variations in ongoing and future work, including explicit reformulations for various risk measures and ambiguity sets, non-affine predicates, non-linear dynamics, infinite-horizon persistent tasks, and alternative risk semantics.

\bibliographystyle{IEEEtran}
\bibliography{references}

\end{document}